\documentclass[11pt,a4paper]{article}

\usepackage{amsfonts,enumerate,subfigure,latexsym,amssymb,amsmath,amsthm} 
\usepackage[algoruled,vlined]{algorithm2e}
\SetKwInOut{Input}{Input}
\SetKwInOut{Output}{Output} 
\usepackage{xspace}
\usepackage[final]{graphicx}
 \def\QQ{{\mathbb Q}}
\def\RR{{\mathbb R}} \def\ZZ{{\mathbb Z}}
\def\A{{\mathcal A}} 
\def\OO{{\mathbb O}}
\def\LP{{\mathbb{LP}}}
\def\OPT{\mbox{OPT}}

\newcommand{\conv}{\mbox{conv}}

\newcommand{\enc}[1]{\langle #1\rangle}

\theoremstyle{plain}
\newtheorem{theorem}{\noindent Theorem}
\newtheorem{corollary}[theorem]{\noindent Corollary} 
\newtheorem{lemma}[theorem]{\noindent Lemma}
\newtheorem{proposition}[theorem]{\noindent Proposition} 
\theoremstyle{definition}
\newtheorem{definition}{\noindent Definition}
\newtheorem{example}{\noindent Example}
\newtheorem{remark}{\noindent Remark}

\title{Efficient edge-skeleton computation for polytopes defined by oracles}
\author{Ioannis\,Z.~Emiris\thanks{Dept Informatics \&
Telecommunications, U.~of Athens, Greece. 
 \texttt{emiris@di.uoa.gr}} 
 \and Vissarion Fisikopoulos\thanks{Computer Science Department, Universit\`e libre de Bruxelles CP 216, Boulevard du Triomphe, 1050 Brussels, Belgium. 
\texttt{vfisikop@ulb.ac.be}} 
 \and Bernd~G\"artner\thanks{Department of Computer Science, ETH Zurich, Switzerland. \texttt{gaertner@inf.ethz.ch}}
}

\begin{document}

\maketitle


\begin{abstract}
In general dimension, there is no known total polynomial algorithm for either
convex hull or vertex enumeration, i.e.\ an algorithm whose complexity depends
polynomially on the input and output sizes.
It is thus important to identify problems (and polytope representations)
for which total polynomial-time algorithms can be obtained.
We offer the first total polynomial-time algorithm for
computing the edge-skeleton (including vertex enumeration) of a polytope
given by an optimization or separation oracle, where
we are also given a superset of its edge directions.
We also offer a space-efficient variant of our algorithm 
by employing reverse search. 
All complexity bounds refer to the (oracle) Turing machine model.
There is a number of polytope classes naturally defined by oracles; 
for some of them neither vertex nor facet representation is obvious.
We consider two main applications, where we obtain
(weakly) total polynomial-time algorithms:
Signed Minkowski sums of convex polytopes, where polytopes can be
subtracted provided the signed sum is a convex polytope, and
computation of secondary, 
resultant, and discriminant polytopes.
Further applications include convex combinatorial optimization
and convex integer programming, where we offer a new approach, thus  
removing the complexity's exponential dependence in the dimension.
\end{abstract}



\section{Introduction}\label{sect:Intro}

Convex polytopes are fundamental geometric objects in science and engineering. Their applications are ranging from theoretical computer science to optimization and algebraic geometry. 
Polytopes in general dimension admit a number of alternative
representations.
The best known, explicit representations for a bounded polytope $P$ are either
the set of its vertices (V-representation)
or a bounded intersection of halfspaces (H-representation). 
Switching between the two representations constitutes the
convex hull and vertex enumeration problems. 
A linear programming problem (LP) on $P$ 
consists in finding a vertex of $P$ that maximizes the inner product 
with a given objective vector $c$. This is very easy if $P$ is in
V-representation; even if $P$ is in H-representation, this LP
can be solved in polynomial time.

In general dimension, there is no polynomial-time algorithm for either 
convex hull or vertex enumeration, since the output size can be
exponential in the worst case by the upper bound theorem~\cite{McMullen}. 
 We therefore wish to take the output
 size into account and say that an algorithm runs in 
{\em total polynomial} time if its
time complexity is bounded by a polynomial in the input \emph{and} output size.
There is no known total polynomial-time algorithm
for either convex hull or vertex enumeration. 
In~\cite{AvisSeidel} they identify, for each known type of convex hull
algorithm, explicit families of polytopes for which the algorithms
run in superpolynomial time.
 
However, finding the vertices of the convex hull of a given point
set reduces to LP
and has thus polynomial complexity in the input.
The algorithm in~\cite{AvisF92} solves, in total polynomial-time,
vertex enumeration for simple polytopes and convex hull 
for simplicial polytopes. 
For 0/1-polytopes a total polynomial-time algorithm for vertex enumeration is presented in~\cite{Bussieck98}, where a 0/1-polytope is such that all 
vertices have coordinates 0 or~1. 
In this paper we establish another case where total polynomial-time 
algorithms exist.

An important explicit representation of a polytope is the {\em edge-skeleton}
(or 1-skeleton), which is the graph of polytope vertices and edges, 
but none of the faces of dimension larger than one. For simple polytopes, the
edge-skeleton determines the complete face lattice (see~\cite{JoswigKaibelK02}
and the references therein), but in general, this is false.
The edge-skeleton is a useful and compact representation employed in
different problems, e.g.\ in computing general-dimensional
Delaunay triangulations of a given pointset:
In~\cite{BoiDevHor09} the authors show how the edge-skeleton suffices
for point location by allowing them to recover only the needed full-dimensional 
simplices of the triangulation.
Another application is in mixed volume computation~\cite{Malajovich14}.

In this paper we study the case where a polytope $P$ is given by an implicit 
representation, where the only access to $P$
is a black box subroutine (oracle) that solves the LP problem on $P$
for a given vector $c$. 
Then, we say that $P$ is given by an {\it optimization}, or {\it LP}
oracle.  Given such an oracle, the entire polytope can be reconstructed,
and both V- and H-representations can be found using the
Beneath-Beyond method \cite{EFKP12,Hug06}, although not in 
total polynomial-time. 

Another important implicit representation of $P$ is obtained through
a {\it separation} oracle (Section~\ref{sect:Oracles}). Celebrated
results of Khachiyan~\cite{Kha1} as well as Gr\"otschel, Lov\'asz and
Schrijver~\cite{GLS93} show that separation and optimization oracles
are polynomial-time equivalent (Proposition~\ref{PropEquival}). Many important
results in combinatorial optimization use the fact that the separation
oracle implies the optimization oracle.
In our study, we also need the other direction:
Given an optimization oracle, compute a separation oracle for $P$. 

The problem that we study is closely related to vertex enumeration. 
We are given an optimization oracle for a polytope $P$ and a set of 
vectors that is guaranteed to contain 
the directions of all edges of $P$; edge directions are given by unit vectors.
We are asked to compute the edge-skeleton of $P$ so
the vertices are also computed.
This is similar to the fundamental Minkowski
reconstruction problem, e.g.~\cite{Gritzmann99onthe},
except that, instead of information on the facets,
we have information about the 1-dimensional faces (and an oracle).
The problem of the reconstruction of a simple polytope by its edge-skeleton 
graph is studied in~\cite{JoswigKaibelK02}.

The most relevant previous work is
an algorithm for vertex enumeration of $P$, given the same input:
an optimization oracle and a superset $D$ of all edge directions
\cite{RothblumOnn} (Proposition~\ref{pr:rothblum_onn}). 
It runs in total polynomial-time in fixed dimension. 
The algorithm computes the zonotope $Z$ of $D$, then computes an 
arbitrary vector in the normal cone of each vertex of $Z$ and calls 
the oracle with this vector.
It outputs all vertices without further information.
Computing the edges from $n$ vertices can be done by $O(n^2)$
calls to LP.  

\subsection{Applications} 

Edge-skeleton computation given an oracle and a superset of edge directions
naturally appears in many applications. 
In Section~\ref{sect:applic} we offer new efficient algorithms for the first two applications below.

One application is the {\em signed Minkowski sum} problem where, besides
addition, we also allow a restricted case of {\em Minkowski subtraction}. 
Let  $A-B$ be polytope $C$ such that $A$ can be written as a sum $A=B+C$. 
In other words, a signed Minkowski sum equality such as 
$P - Q + R - S = T$ should be interpreted as $P + R = Q + S + T$.
Such sums are motivated by the fact that resultant and discriminant
polytopes (to be defined later) are written as signed sums of
secondary polytopes \cite{MicCoo00}, \cite[Thm~11.1.3]{GKZ}.  
Also, matroid polytopes and generalized permutahedra can be written as signed Minkowski sums~\cite{Ardila10}.
We assume that the summands are given by optimization oracles and the supersets of their edge directions. 
This is natural since we show that these supersets can be precomputed for resultant and secondary polytopes.    

Minkowski sums have been studied extensively.
Given $r$ V-polytopes in $\RR^d$, Gritzmann et al.~\cite{Gritzmann93} 
deal with the various Minkowski sum problems
that occur if they regard none, one, or both of $r$ and $d$ as constants.
They give polynomial algorithms for fixed $d$ regardless of the input 
representation. For varying $d$ they show that no polynomial-time algorithm
exists except for the case of fixed $r$ in the binary model of computation.
Fukuda~\cite{Fukuda04} (extended in~\cite{FukudaW05})
gives an LP-based 
algorithm for the Minkowski sum of polytopes in V-representation 
whose complexity, in the binary model of computation,
is total polynomial, and depends polynomially on $\delta$, 
which is the sum of the maximum vertex degree in each summand.
However, we are not aware of any algorithm for signed Minkowski sums
and it is not obvious how the above algorithms for Minkowski sums 
can be extended to the signed case. 

The second application is resultant, secondary as well as 
discriminant polytopes.
For resultant polytopes at least, the only plausible representation
seems to be via optimization oracles~\cite{EFKP12}.
{\em Resultants} are fundamental in computational algebraic geometry
since they generalize determinants to nonlinear systems~\cite{St94,GKZ}.
The Newton polytope $R$ of the resultant, or {\em resultant polytope},
is the convex hull of the exponent vectors corresponding to nonzero terms.
A resultant is defined for $k+1$ pointsets in $\ZZ^k$.
If $R$ lies in $\RR^d$, the total number of input points is $d+2k+1$.
If $n$ is the number of vertices in $R$, typically $n\gg d\gg k$, 
so $k$ is assumed fixed.
A polynomial-time optimization oracle yields
an output-sensitive algorithm for the computation of $R$ \cite{EFKP12} (Lemma~\ref{EFKP12_prop}).

This approach can also be used for 
computing the secondary and discriminant polytopes, defined in~\cite{GKZ};
cf.~\cite{DeLRamSan} on secondary polytopes.
The secondary polytope of a pointset is a fundamental object since it
offers a realization of the graph of regular triangulations
of the pointset.
A total polynomial-time algorithm for the secondary polytope 
when $k$ is fixed is given in~\cite{Masada96}.
A specific application of discriminant polytopes is discussed
in~\cite{discrim_vol}, where the author establishes
a lower bound on the volume of the discriminant polytope of a
multivariate polynomial, thus refuting a conjecture by E.I.~Shustin on 
an asymptotic upper bound for the number of real hypersurfaces.

The size of all these polytopes is typically exponential in $d$: the number of
vertices of $R$ is $O(d^{2d^2})$~\cite{St94}, and the number of $j$-faces 
(for any $j$) of the secondary polytope 
is $O(d^{(d-1)^2})$, which is tight if $d$ is fixed~\cite{Billera1990155}.  

More applications of our methods exist.  
One is in \emph{convex combinatorial optimization}: given
$\mathcal{F}\subset 2^N$ with $N=\{1,\dots,n\}$, a vectorial weighting
$w: N \rightarrow\QQ^d$, and a convex functional $c: \QQ^d \rightarrow
\QQ$, find $F\in\mathcal{F}$ of maximum value $c(w(F))$. This
captures a variety of (hard) problems studied in operations research 
and mathematical programming, including quadratic assignment, 
scheduling, reliability, bargaining games, and inventory management, 
see~\cite{OnnR04} and references therein. 
The standard linear combinatorial optimization problem is the special 
case with $d=1, w:N \rightarrow\QQ$,
and $c:\QQ \rightarrow\QQ:x\mapsto x$ being the identity.
As shown in~\cite{OnnR04}, a convex combinatorial optimization problem
can be solved in polynomial-time for fixed $d$, if we know the edge
directions of the polytope given by the convex hull of the incidence
vectors of the sets in $\mathcal{F}$.

Another application is \emph{convex integer maximization}, where we
maximize a convex function over the integer hull of a
polyhedron.  In~\cite{DeLoera20091569}, the vertex enumeration
algorithm of~\cite{RothblumOnn}---based on the knowledge of edge
directions---is used to come up with polynomial algorithms for many
interesting cases of convex integer maximization, such as multiway
transportation, packing, vector partitioning and clustering.  A set
that contains the directions of all edges is computed via Graver
bases, and the enumeration of all vertices of a projection of the
integer hull suffices to find the optimal solution.

\subsection{Our contribution} 
We present the first total polynomial-time algorithm for computing the edge-skeleton
of a polytope, given an optimization oracle, and a set of directions 
that contains the polytope's edge directions. 
The polytope is assumed to have some (unknown) H-representation with an
arbitrary number of inequalities, but each of known bitsize, 
as shall be specified below.
Our algorithm also works if the polytope is given by a separation oracle.
All complexity bounds refer to the (oracle) Turing machine model, thus
leading to (weakly) polynomial-time algorithms when the oracle is
of polynomial-time complexity. 
By employing  the reverse search method of~\cite{AvisF92} we offer 
a space-efficient variant of our algorithm.  It remains open
whether there is also a strongly polynomial-time algorithm 
in the real RAM model of computation.

Our algorithm yields the first (weakly) total polynomial-time algorithms 
for the edge-skeleton (and vertex enumeration) of
signed Minkowski sum, and resultant polytopes (for fixed $k$). 
For both polytope classes, optimization oracles are naturally and efficiently
constructed, whereas it is not straightforward to
obtain the more commonly employed membership or separation oracles. For signed Minkowski sum we assume that we know the supersets of edge directions for summands. This includes the important cases where the summands are V-polytopes, and secondary  polytopes. 
For resultant polytopes, optimization oracles offer
the most efficient known representation.
Our results on resultant polytopes extend to secondary polytopes, as well as discriminant polytopes.
Recall that a different approach in
the same complexity class is known for secondary polytopes~\cite{TOPCOM2}.

Regarding the problems of convex combinatorial optimization and 
convex integer programming the current approaches use the algorithm 
of~\cite{RothblumOnn} whose complexity has an exponential dependence 
on the dimension (Proposition~\ref{pr:rothblum_onn}).
The utilization of our algorithm instead offers an alternative approach while 
removing the exponential dependence on the dimension. 

\paragraph{Outline}
The next section specifies our theoretical framework.
Section\ \ref{sect:algorithms} introduces polynomial-time algorithms for
the edge-skeleton.
Section\ \ref{sect:applic} applies our results to signed Minkowski sums, 
as well as resultant and secondary polytopes.
We conclude with open questions.  

\section{Well-described polytopes and oracles}\label{sect:Oracles} 

This section describes our theoretical framework and relates
the most relevant oracles. We start with the notation used in this paper followed by some basics from polytope theory; for a detailed presentation we refer to~\cite{Ziegler}.

We denote by $d$ the ambient space dimension
and $n$ the number of vertices of the output (bounded) polytope; $k$ denotes
dimension when it is fixed (e.g.\ input space for resultant
polytopes); conv$(A)$ is the convex hull of $A$.
Moreover, $\varphi$ is an upper bound for the encoding length
of every inequality defining a well-described polytope (see the next section); 
$\enc{X}$ denotes the binary
encoding size of an explicitly given object $X$ (e.g., a set of
vectors). For a well-described and implicitly given polytope
$P\subseteq\RR^d$, we 
will define $\enc{P}:=d+\varphi$.
Let $\OO:\RR\rightarrow\RR$ denote a polynomial such that the
oracle conversion algorithms of Proposition~\ref{PropEquival}
all run in oracle polynomial-time $\OO(\enc{P})$
for a given well-described polytope $P$. The polynomial
$\LP:\RR\rightarrow\RR$ 
is such that $\LP(\enc{A} + \enc{b} + \enc{c})$ bounds the runtime of 
maximizing $c^Tx$ over the polyhedron $\{x\ |\ Ax\leq b\}$.

A convex polytope $P\subseteq\RR^d$ can be represented as the convex hull
of a finite set of points, called the \emph{V-representation} of $P$. 
In other words, $P=$ conv$(A)$, where $A=\{p_1,\dots,p_n\}\subseteq\RR^d$. 
Another, equivalent representation of $P$ is as the bounded  intersection of a finite set of halfspaces or linear inequalities, called the \emph{H-representation} of $P$.
That is, $P=\{x\,|\,Ax\leq b\},A\subseteq\RR^{m\times d},x\in\RR^d,b\in\RR^m$. 
Given $P$, an inequality or a halfspace $\{a^Tx\leq \beta\}$ (where
$a\in\RR^d,\beta\in\RR$) is called \emph{supporting} if $a^Tx\leq
\beta$ for all $x\in P$ and $a^Tx=
\beta$ for some $x\in P$.
The set $\{x\in P\ |\ a^Tx= b\}$ is a \emph{face} of $P$.

\begin{definition}
The {\it polar dual polytope} of $P$ is defined as:
$$
P^*:=\{c\in\RR^d : c^Tx\leq 1 \text{~for all}\ x\in P\}\subseteq\RR^d,
$$ 
where we assume that the origin $\mathbf{0}\in int(P)$,
the relative interior of $P$,
i.e.\ $\mathbf{0}$ is not contained in any face of $P$ of dimension $< d$.
\end{definition}

For our results, we need to assume that the output
polytope is \emph{well-described}~\cite[Definition 6.2.2]{GLS93}. This
will be the case in all our applications. 

\begin{definition}
A rational polytope $P\subseteq\RR^d$ is
{\em well-described} (with a parameter $\varphi$ that we need to know
explicitly) if there exists an H-representation for $P$ in which
every inequality has encoding length at most $\varphi$.
The 
\emph{encoding length} of a well-described polytope is 
$\enc{P} = d + \varphi$. 
Similarly, the \emph{encoding length} of a set of vectors $D\subseteq\RR^d$ 
is $\enc{D} = d + \nu$ if every vector in $D$ 
has encoding length at most $\nu$.
\end{definition}

In defining $P$, the inequalities are not known themselves, and
we make no assumptions about their number.
The following lemma connects the encoding length of inequalities 
with the encoding length of vertices.

\begin{lemma}{\rm \cite[Lemma 6.2.4]{GLS93}}\label{lem:encoding}
 Let $P\subseteq\RR^d$. If every inequality in an H-re\-presentation 
for $P$ has encoding length at most $\varphi$, then every vertex of $P$ 
has encoding length at most $4d^2\varphi$. 
If every vertex of $P$ has encoding length at most $\nu$, then every
inequality of its H-representation has encoding length at most $3d^2\nu$.
\end{lemma}

The natural model of computation when $P$ is given by an oracle is
that of an \emph{oracle Turing machine}~\cite[Section 1.2]{GLS93}.
This is a Turing machine that can (in one step) replace any input to
the oracle (to be prepared on a special oracle tape) by the output resulting
from calling the oracle, where we assume that the output size is
polynomially bounded in the input size. An algorithm is \emph{oracle
  polynomial-time} if it can be realized by a polynomial-time oracle
Turing machine. Moreover it is {\em total polynomial-time} if its
time complexity is bounded by a polynomial in the input \emph{and} output size.

In this paper, we consider three oracles for polytopes; they can more
generally be defined for (not necessarily bounded) polyhedra,
but we do not need this:
\begin{itemize}
\item {\it Optimization} ($\text{OPT}_P(c)$):
 Given vector $c\in \RR^d$, either find a point $y\in P$
maximizing $c^{T}x$ over all $x\in P$, or assert $P=\emptyset$.
\smallskip

\item {\it Violation} ($\text{VIOL}_P(c,\gamma)$): Given vector $c\in\RR^d$ and
$\gamma\in\RR$, either find point $y\in P$ such that $c^Ty>\gamma$, or 
assert that $c^Tx\leq\gamma$ for all $x\in P$. \smallskip

\item {\it Separation} (SEP$_P(y)$):
 Given point $y\in \RR^d$, either certify that $y\in P$, or
find a hyperplane that separates $y$ from $P$; i.e.\
find vector $c\in\RR^d$ such that $c^Ty> c^Tx$ for all $x\in P$.  \smallskip
\end{itemize}

The following is a main result of~\cite{GLS93}
and the cornerstone of our method.

\begin{proposition} {\rm\cite[Theorem~6.4.9]{GLS93}}\label{PropEquival}
For a well-described polytope, any one of the three
aforementioned oracles is sufficient to compute the other
two in oracle polynomial-time. The runtime (polynomially) depends on
the ambient dimension $d$ and the bound $\varphi$ for the maximum
encoding length of an inequality in some H-representation of $P$. 
\end{proposition} 

For applications in combinatorial optimization, an extremely important
feature is that the runtime does not depend on the number of inequalities
that are needed to describe $P$. Even if this number is exponential in
$d$, the three oracles are polynomial-time equivalent.

Another important corollary is that linear programs can be solved in
polynomial-time. Indeed, an explicitly given (bounded coefficient) system
$Ax\leq b, x\in\RR^d$ of inequalities defines a well-described polytope $P$,
for which the separation oracle is very easy to implement in time
polynomial in $\enc{P}$; hence, the oracle polynomial-time algorithm for
$\text{OPT}_P(c)$ becomes a (proper) polynomial-time algorithm.

\section{Computing the edge-skeleton}\label{sect:algorithms}

This section studies total polynomial time algorithms for the edge-skeleton of a polytope.
We are given a well-described polytope $P\subseteq\RR^d$ via an
optimization oracle $\text{OPT}_P(c)$ of $P$. Moreover, we are given a
superset $D$ of all edge directions of $P$; to be precise, we define
\[
D(P):=\left\{\frac{v-w}{\|v-w\|}:
	\mbox{$v$ and $w$ are adjacent vertices of $P$}\right\}
\]
to be the set of (unit) edge directions, and we assume that for every
$e\in D(P)$, the set $D$ contains some positive multiple $te, t\in\RR,
t>0$. Slightly abusing notation, we write $D\supseteq D(P)$.

The goal is to efficiently compute the edge-skeleton of $P$, i.e.\ its
vertices and the edges connecting the vertices. Even if $D=D(P)$, this
set does not, in general, provide enough information for the task, so
we need additional information about $P$; here we assume an
optimization oracle. 

Vertex enumeration with this input has been studied in the real RAM
model of computation where we count the number of arithmetic operations:

\begin{proposition}{\rm \cite{RothblumOnn}}\label{pr:rothblum_onn}
Let $P\subseteq \RR^d$ be a polytope given by OPT$_P(c)$,
and let $D\supseteq D(P)$ be a superset of the edge directions of $P$.
The vertices of $P$ are computed using
$O(|D|^{d-1})$ arithmetic operations and 
$O(|D|^{d-1})$ calls to $\text{OPT}_P(c)$.
\end{proposition}

If $P$ has $n$ vertices, then $|D(P)|\leq \binom{n}{2}$, and this
is tight for neighborly polytopes in general position~\cite{Ziegler}. This
means that the bound of Proposition~\ref{pr:rothblum_onn} is
$O\left(n^{2d-2}\right)$, assuming that $|D|=\Theta(|D(P)|)$.

We show below that the edge-skeleton can be computed in oracle
total polynomial-time for a well-described polytope, which possesses an
(unknown) H-representation with encoding size
$\varphi$.  Thus, we show that the
exponential dependence on $d$ in Proposition~\ref{pr:rothblum_onn} can be
removed in the Turing machine model of computation, leading to a
(weakly) total polynomial-time algorithm. It remains open whether there is
also a strongly total polynomial-time 
algorithm with a total polynomial runtime
bound in the real RAM model of computation.

\begin{algorithm}[t]
 \BlankLine
  \Input{Optim. oracle $\text{OPT}_P(c)$,
superset $D$ of edge directions $D(P)$} 
  \Output{The edge-skeleton (and vertices) of $P$}
  \BlankLine\BlankLine
  Compute some vertex $v_0\in P$\;
  $V_{P}\leftarrow \emptyset;\, W\leftarrow \{v_0\}$;\, $E_{P}\leftarrow \emptyset$\;
  \While{$W\neq\emptyset$}{
    Choose the next element $v\in W$ and remove it from $W$\;
    $V_{P}\leftarrow V_{P}\bigcup \{v\}$\;
    $V_{cone}\leftarrow \emptyset$\;
    \ForEach{$e\in E$}{
      $w\leftarrow\mathrm{argmax}\{v+te\in P, t\geq 0\}$\;
      \If{$w\neq v$}{
         $V_{cone}\leftarrow V_{cone}\bigcup \{w\}$\;
      }
    }
    Remove non-vertices of $P$ from $V_{cone}$\; 
    \ForEach{$w\in V_{cone} $}{
      \lIf{$w\notin V_{P}$}{$W\leftarrow W\bigcup\{w\}$}
      \lIf{$\{v,w\}\notin E_{P}$}{$E_P\leftarrow E_P\bigcup \{v,w\}$}
    }
  } 
  \Return $V_P,E_P$\;
  \BlankLine
  \caption{\label{Alg:edge_skeleton} Edge\_Skeleton\, ($\text{OPT}_P, D$)}
\end{algorithm} 

The algorithm (Algorithm~\ref{Alg:edge_skeleton}) is as follows.
 Using $\text{OPT}_P(c)$, we find some
vertex $v_0$ of $P$ (this can be done even if
$\text{OPT}_P(c)$ does not directly return a vertex
\cite[Lemma~6.51]{GLS93}, \cite[pp.~255--256]{EPL82}).

We maintain sets $V_P, E_P$ of vertices and their incident edges,
along with a queue $W\subseteq V_P$ of the vertices for which we have
not found all incident edges yet. Initially, $W=\{v_0\},
V_P=E_P=\emptyset$.
When we process the next vertex $v$ from the queue, it remains to find
its incident edges: equivalently, the neighbors of $v$. 
To find the neighbors, we first build a set $V_{cone}$ of candidate
vertices. We know that for every neighbor $w$ of $v$, there must be an
edge direction $e$ such that $w=v+te$ for suitable $t>0$.
More precisely, $w$ is the point corresponding to maximum $t$ in the
$1$-dimensional polytope $Q(e):=P\cap\{x\ |\  x=v+te, t\geq 0\}$, 
where the latter equals the intersection of $P$ with the
ray in direction $e$ and apex at $v$.
Hence, by solving $|D|$ linear programs, one for
every $e\in D$, we can build a set $V_{cone}$ that is guaranteed
to contain all neighbors of $v$. To solve these linear programs,
we need to construct optimization oracles for $Q(e)$. To do this,
we first construct $\text{SEP}_P(y)$ from $\text{OPT}_P(c)$ in oracle
polynomial-time according to~Proposition~\ref{PropEquival}. Thus,
the construction of $\text{SEP}_{Q(e)}(y)$ is elementary, and since
also $Q(e)$ is well-described, we can obtain
$\text{OPT}_{Q(e)}(c)$ in oracle polynomial-time.

In a final step, we remove the candidates that do not yield neighboring
vertices. For this, we first solve a linear program to compute a
hyperplane $h$ separating $v$ from the candidates in $V_{cone}$; since
$V_{cone}$ is a finite subset of $P\setminus\{v\}$, such a hyperplane
exists, and w.l.o.g.\ $v=\mathbf{0}$ and $h=\{x\ |\ x_d=1\}$. Let $C$ be the cone
generated by the set $V_{cone}$. We compute the extreme points of
$C\cap\{x_d=1\}$, giving us the extremal rays of $C$.  Finally, we
remove every point from $V_{cone}$ that is not on an extremal ray.

The correctness of the algorithm relies on the following Lemma.

\begin{lemma}\label{lem:extreme_rays}
  Let $v$ be a vertex of $P$ processed during
  Algorithm~\ref{Alg:edge_skeleton}, where we assume w.l.o.g.\ that $v=\mathbf{0}$
  and the set $V_{cone}$ of candidates is separated from $v$ by the
  hyperplane $\{x\ |\  x_d=1\}$.

  A point $w\in \RR^d$ is a neighbor of $v$ if and only if $w$ is on
  some extremal ray of the cone $C$ generated by $V_{cone}$. Here, an
  extremal ray is a ray whose intersection with the hyperplane
  $\{x_d=1\}$ is an extreme point of the polytope $C\cap\{x\ |\ x_d=1\}$.
\end{lemma}

\begin{proof}
  Suppose that $w$ is a neighbor of $v$. By construction, 
  $w\in V_{cone}$. Moreover, since $\{v,w\}$ is an edge, there is a
  supporting hyperplane $h=\{a^Tx=0\}$ (recall that $v=\mathbf{0}$) such that
  $a^Tx=0$ for all $x\in\conv(\{v,w\})$ and $a^Tx>0$ for all
  $p\in P\setminus\conv(\{v,w\})$. For each $q\in V_{cone}$, let
  $c(q) =\frac{1}{q_d} q \in C\cap\{x\ |\ x_d=1\}$.  We have $q_d>0$ by 
  construction. Furthermore, $a^Tc(w)=0$ while $a^Tc(q)>0$
  for $q\in V_{cone}$, unless $q\in\conv(\{v,w\})$. In the latter 
  case, $c(q)=c(w)$. Hence, $c(w)$ is the only point $y$ 
  in $C\cap\{x_d=1\}$ such that $a^Ty=0$, and this implies
  that $c(w)$ is an extreme point of $C\cap\{x_d=1\}$. So $w$
  is on some extremal ray of $C$.

  For the other direction, suppose that $w\in V_{cone}$ is 
  on the extremal ray $\{te\ |\ t\in\RR\}$. So $c(w)$ is an extreme
  point of $C\cap\{x\ |\ x_d=1\}$. This means, there exists a vertical
  hyperplane $h=\{a^Tx = \beta\}$ with $a_d=0$ such that $a^Tc(w)=\beta$, and
  $a^Tc(q)>\beta$, for all $q\in V_{cone}$ satisfying $c(q)\neq c(w)$.
  Now define the hyperplane $\overline{h}=\{\overline{a}^Tx =
  0\}$ with $\overline{a}=(a_1,\ldots,a_{d-1},-\beta)$. It follows that
  $\overline{a}^Tq\geq 0$ for all $q\in V_{cone}$, so the positive halfspace of
  $\overline{h}$ contains $C$ and thus also $P$ since $P\subseteq C$. 
  We claim that $\overline{h}\cap P=\conv(\{v,w\})$, which proves that
  $\conv(\{v,w\})$ is an edge of $P$ and hence $w$
  is a neighbor of $v$. 

  To see this, we first observe that $\overline{a}^Tw=0$ and
  $\overline{a}^Tq>0$ for all $q\in V_{cone}$ that are not multiples
  of $w$, so $\overline{h}\cap P\subseteq \overline{h}\cap C = \{te\
  |\ t\in\RR\}$. On the other hand, we know from the construction of
  $V_{cone}$ that $w$ is the highest point of $P$ (the one with
  maximum $t$) on the ray $\{te\ |\
  t\in\RR\}$, so we indeed get $\overline{h}\cap P=\conv(\{v,w\})$.
\end{proof}

We now bound the {\it time complexity} of Algorithm~\ref{Alg:edge_skeleton}.

\begin{theorem}\label{thm:edge-skeleton}
Given $OPT_P$ and a superset of edge directions $D$ of a
well-described polytope $P\subseteq\RR^d$ with $n$ vertices, and $m$ edges 
Algorithm~\ref{Alg:edge_skeleton} computes the edge-skeleton of $P$ in  
oracle total polynomial-time 
\[
O\left( n|D|\left(\OO(\enc{P}+\enc{D}) +
	\LP( d^3|D| \, (\enc{P}+\enc{D})) \, +  d \log n \right)  \right),
\]
where $\enc{D}$ is the binary encoding length of the vector set $D$.
\end{theorem} 

\begin{proof} 
We call $\text{OPT}_P(x)$ to find the first vertex of $P$. Then, there
are $O(n)$ iterations. In each one, we construct $O(|D|)$ oracles for
polytopes $Q(e)$ of encoding length at most $\enc{P}+\enc{D}$. 
We also
compute the (at most $n$) extreme points from a set of at most $|D|$
candidate points. This can be done by solving $|D|$ linear programs
whose inequalities have coefficients that are in turn coordinates of
vertices of the $Q(e)$'s. By~Lemma~\ref{lem:encoding}, these
coordinates have encoding lengths bounded by $4d^2(\enc{P}+\enc{D})$,
and the number of coefficients in each linear program is $O(|D|d)$.
At each vertex we have to test whether the computed vertices and edges 
are new. In the course of the algorithm these tests are at most $O(m)=O(n|D|)$, 
where $m$ the number of $P$ edges. We can test whether a vertex (or an edge) is new in $O(d\log n)$.
\end{proof}

\subsection{Reverse search for edge-skeleton.} 
We define a reverse search procedure based on~\cite{AvisF92} to optimize the 
space used by Algorithm~\ref{Alg:edge_skeleton}. 
Given a vertex of $P$, the set of adjacent edges can be constructed 
as described above. 
Then we need to define a total order over the vertices of the polytope. 
Any generic vector $c\in\RR^d$ induces such an order on the vertices, 
i.e.\ the order of a vertex $u$ is that of $c^Tu$. 
In other words, we can define a reverse search tree on $P$
with root the vertex $v$
that maximizes $c^Tv$ over all the vertices of $P$, where $c$ is
the vector given to OPT$_P$ for initializing $P$.
Technically, the genericity assumption on $c$ can be avoided 
by sorting the vertices w.r.t.\ the lexicographical ordering 
of their coordinates.

Reverse search also needs an {\em adjacency procedure} which,
given a vertex $v$ and an 
integer $j$, returns the $j$-th adjacent vertex of $v$, as well as 
a {\em local search procedure} allowing us to move from any vertex to its 
optimal neighbor w.r.t.\ the objective function. 
Both procedures can be implemented by computing all the adjacent vertices
of a given vertex of $P$ 
as described above, and then returning the best (or the $j$-th)
w.r.t.\ the ordering induced by $c$.

The above procedures can be used by a reverse search variant of 
Algorithm~\ref{Alg:edge_skeleton} that traverses (forward and backward) 
the reverse search tree while keeping in memory only a constant 
number of $P$ vertices and edges. 
On the contrary, both the original Algorithm~\ref{Alg:edge_skeleton} and  
the algorithm of Proposition~\ref{pr:rothblum_onn} 
need to store all vertices of $P$ whose number is exponential in $d$ 
in the worst case. 
Note that any algorithm should use space at least $O(d|D|)$ to store the input 
set of edge directions.  
The above discussion yields the following result 
(encoding length of $P$ vertices comes from Lemma~\ref{lem:encoding}).

\begin{corollary}\label{cor:edge-skeleton-rs}
Given $OPT_P$ and a superset of edge directions $D$, a variant of 
Algorithm~\ref{Alg:edge_skeleton} that uses reverse search 
runs in additional to the input space $O(4d^2\enc{P}+\enc{D})$ 
while keeping the same asymptotic time complexity.
\end{corollary} 

\section{Applications}\label{sect:applic} 

This section studies the performance of Algorithm~\ref{Alg:edge_skeleton} in certain classes of polytopes where we do not assume that we know the set of edge directions a priori. To this end, we describe methods for pre-computing a (super)set of the edge directions. 
 
We start by describing the computation of the set of edge-directions in arbitrary polytopes using the formulation of standard polytopes.
Then we study two important classes of polytopes where the number of edge directions can be efficiently precomputed and thus provide new, total polynomial-time algorithms for their representation by an edge-skeleton. In particular, we study 
signed Minkowski sums, and resultant and secondary polytopes.
We show that these polytopes are well-described
and naturally defined by optimization 
oracles, which provide a compact representation.  

\subsection{Standard polytopes}
First we discuss the performance of Algorithm~\ref{Alg:edge_skeleton} on general polytopes. 
Any convex polytope $P=\{x\ |\ Ax\leq b\}$ can be written as a linear projection of a polytope $Q=\{(x,y)\ |\ Ax+Iy=b, y\geq 0\}$, where $A\subseteq\RR^{m\times d}$ and we introduce the slack variables $y\in\RR^m$ and $P=\pi(Q)$, by the linear mapping $\pi(x,y)=x$. We can rewrite $Q$ as the so-called {\it standard polytope} $\{x'\ |\ Bx'=b, x'\geq 0\}$. The set $E$ of edge directions of $Q$ is covered by the set of circuits of $B$ (Cf.\ Lemma 2.13.\ in~\cite{OnnR04}). Moreover, each edge direction of $P$ is the projection of some direction in $E$ under the mapping $\pi$ (Cf.\ Lemma 2.4.\ in~\cite{OnnR04}). However, 
the number of circuits of $B$ will be exponential in $d$.   
On the other hand, for small dimensions Algorithm~\ref{Alg:edge_skeleton} could be an efficient choice for edge-skeleton computation or vertex enumeration. 

\subsection{Signed Minkowski sums} 
Recall that the {\it Minkowski sum} of (convex) polytopes $A,B\subseteq\RR^d$ is defined as 
$$
A+B:=\{a+b\ |\ a\in A,b\in B\}.
$$
Following~\cite{Schneider93} the {\it Minkowski subtraction} is defined as 
$$A-B:=\{x\in\RR^d\ |\ B+x\subseteq A\}.$$ 
Here we consider a special case of Minkowski subtraction where $B$ is a summand of $A$. Equivalently, if $A-B=C$ then $A=B+C$.
A {\em signed Minkowski sum} combines Minkowski sums and subtractions,
namely
$$
P = s_1P_1+ s_2P_2 + \dots+ s_rP_r, \; s_i\in\{-1,1\},
$$
where all $P_i$ are convex polytopes and so is $P$.

We also define the sum (or subtraction) of two optimization oracles as the Minkowski 
sum (or subtraction) of the resulting vertices. In particular, if OPT$_P(c)=v$ and 
$\OPT_{P'}(c')=v'$ for $v,v'$ vertices of $P,P'$ respectively, then $\OPT_P(c)+\OPT_{P'}(c)=v+v'$ and $\OPT_P(c)-\OPT_{P'}(c)=v-v'$.
An optimization oracle for the signed Minkowski sum is given by the signed sum of the optimization oracles of the summands.

\begin{lemma}\label{lem:Mink_sum_oracle}
If $P_1,\dots,P_r\subset\RR^d$ are given by optimization oracles, 
then we compute an optimization oracle
for signed Minkowski sum $P=\sum_{i=1}^r s_iP_i$ in $O(r)$.
\end{lemma}
\begin{proof} 
Assume w.l.o.g.\ that $s_1=\cdots=s_k=1\ne s_{k+1}=\cdots=s_r=-1$.
Then, given $P=\sum_{i=1}^r s_iP_i$ we have 
$P+\sum_{i=k+1}^r P_i=\sum_{i=1}^k P_i=P'$. 
Let OPT$_{P'}(c)=v$ for some vertex $v$ of $P'$ and vector $c\in\RR^d$. 
It suffices to show that 
\begin{align*}
\mbox{OPT}_{P'}(c)=v&=v_1+\dots+v_k=\sum_{i=1}^k\mbox{OPT}_{P_i}, 
\end{align*} 
which follows from Minkowski sum properties: $v=v_1+\dots+v_k$ for vertices 
$v_i$ of $P_i$ and $norm_P(v)\subseteq norm_{P_i}(v_i)$, for $i=1\dots k$. 
Here $norm_P(v)$ denotes the {\it normal cone} of vertex $v$ of $P$, i.e.\ the set of all vectors $c$ such  that $c^Tx\leq c^Tv$ for all $x\in P$.
Therefore, we can compute OPT$_P$ with $r$ oracle calls to OPT$_{P_i}$ 
for $i=1,\dots,r$.
This yields a complexity of $O(r)$ for $\OPT_P$ since,
by definition of oracle polynomial-time, the oracle calls
in every summand are of unit cost.
\end{proof}

\begin{figure}[t]\centering
\includegraphics[width=\textwidth]{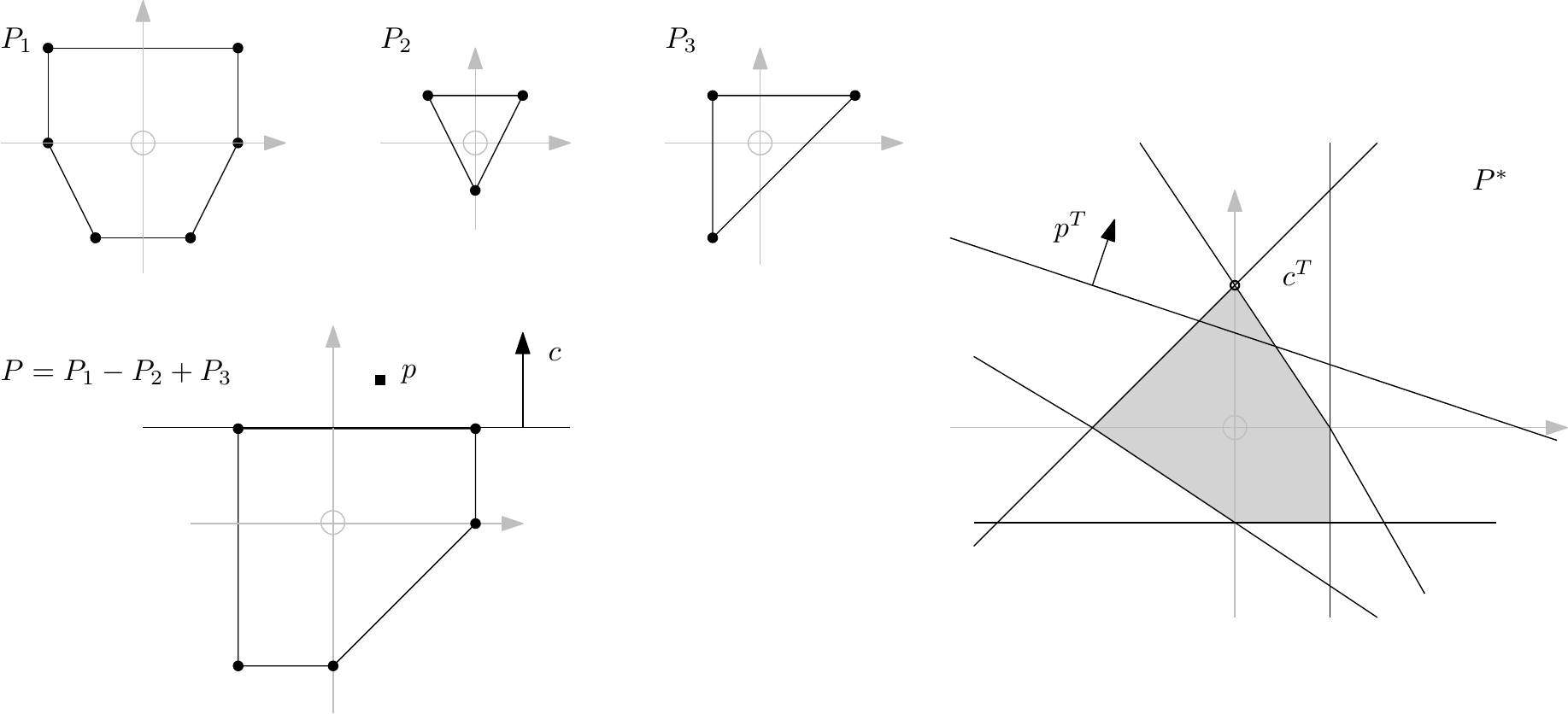}
\caption{\small Signed Minkowski sum oracles.
\label{fig:Mink_sum}} \end{figure}

\begin{example}\label{eg:minksum}
Here we illustrate the above definitions and constructions as well as the standard reductions from~\cite{GLS93}.
Consider the polytopes $P_1,P_2,P_3$, their signed Minkowski sum
$P=P_1-P_2+P_3$, and its polar $P^*$ as shown in Figure~\ref{fig:Mink_sum}.
Observe that $P_1=P_2+ S$, where $S$ is a square.
Assume that $P_1,P_2,P_3$ are given by $\text{OPT}_{P_1},\text{OPT}_{P_2},$ $\text{OPT}_{P_3}$ oracles. 

Then, $\text{OPT}_P(c)=\text{OPT}_{P_1}(c)-\text{OPT}_{P_2}(c)+\text{OPT}_{P_3}(c)$ for some vector $c\in\RR^d$.
If $P$ satisfies the requirements of Proposition~\ref{PropEquival} then,
having access to OPT$_P(c)$, we compute $\text{SEP}_P(p)$ in oracle 
polynomial-time for point $p\in\RR^d$.
In particular, asking if $p\in P$ is equivalent to asking if 
$H:= \{x\ |\ p^Tx\leq 1\}$ is a valid inequality for $P^*$. The latter 
can be solved by computing the point $c^T$ in $P^*$ that maximizes the inner 
product with the outer normal vector of $H$ and test if it validates $H$.
If this happens then $\text{SEP}_P(p)$ returns that $p\in P$, otherwise it returns 
$p\notin P$ with separating hyperplane $\{x\ |\ c x=1\}$.
\end{example}

Let $n$ denote the number of vertices of $P$.
An oracle for $P$ is provided by Lemma~\ref{lem:Mink_sum_oracle}.
Then, the entire polytope can be reconstructed,
and both V- and H-representations can be found by Proposition~\ref{BB}. 

\begin{proposition}{\rm\cite{EFKP12}}\label{BB}
Given OPT$_P$ for $P\subseteq \RR^d$, its V- and H-representations
as well as a triangulation $T$ of $P$ can be computed in
$$
O(d^5n|T|^2)\text{ arithmetic operations, and } O(n+f) \text{ calls to OPT}_P,
$$
where $n$ and $f$ are the number of vertices and facets of $P$, respectively,
and $|T|$ the number of $d$-dimensional simplices of $T$.
\end{proposition}

\begin{corollary}
  Given optimization oracles for $P_1,\dots, P_r\subseteq \RR^d$, 
  we construct the V- and H-representations, and a triangulation $T$ of signed Minkowski sum 
  $P=P_1+s_2P_2 + \cdots + s_rP_r, \, s_i\in\{-1,1\}$
  in output sensitive complexity, namely
  $O(d^5n|T|^2+(n+f)r)$, where $n,f$ are the number of
  vertices and facets in $P$ and $|T|$ the number of full-dimensional simplices of $T$.
\end{corollary}

In the above complexity the number of $d$-dimensional simplices of the computed triangulation $T$ can be exponential in $d$ which is
essentially given by the Upper Bound Theorem for spheres, i.e.\ $|T|=O(n^{(d+1)/2})$~\cite{stanley2004combinatorics}. This stresses the need for total polynomial-time algorithms for the edge-skeleton of $P$.
Note that it is not assumed that the polytopes are well-described. 
But we assume the input contains a superset of all edges for each $P_i$.
In one of the most important cases where we are given the vertices of all summands $P_i$, 
we can compute all edges in each $P_i$ by solving a linear 
program for each pair of vertices.
Each such pair defines a candidate edge.  
Hence, the overall computation of the edges of $P_i$'s is polynomial.  

\begin{corollary}\label{cor:Mink_sum_sep}
Given optimization oracles for well-described $P_1,\dots, P_r\subseteq \RR^d$,  
and supersets of their edge directions $D_1,\dots,D_r$,    
the edge-skeleton of the signed Minkowski sum $P$ can be computed in
oracle total polynomial-time by Algorithm~\ref{Alg:edge_skeleton}.
\end{corollary}

\begin{proof}
To be able to apply Algorithm~\ref{Alg:edge_skeleton}, 
first we should show that $P$ is well-described. Let $\enc{P_{max}}$ be the 
maximum encoding length of summands $P_1,\dots, P_r$. 
Then by~Lemma~\ref{lem:encoding}, 
the encoding length of the coordinates of summand vertices is $4d^2\enc{P_{max}}$. 
Thus,  $4d^2\enc{P_{max}}+\enc{r}$ is the encoding 
length of the coordinates of $P$ vertices. 
Finally, $\enc{P}=d+12d^4 \enc{P_{max}}+3d^2\enc{r}$ by~Lemma~\ref{lem:encoding}.
Now $\text{OPT}_P$ is computed by Lemma~\ref{lem:Mink_sum_oracle} in $O(r)$.
The superset of the edge directions of $P$ is $D=\bigcup_{s_i>0}D_i$,
because $D(P_1-P_2)\subset D(P_1)$ since
$P_1-P_2=P_3\Leftrightarrow P_1=P_2+P_3$.
\end{proof}

Our algorithm assumes that, in the Minkowski subtraction $A-B$,
$B$ is a summand of $A$ and does not verify this assumption.

\subsection{Secondary and resultant polytopes}

The {\em secondary polytope} $\Sigma$ 
of a set of $d$ points $A=\{p_1,\dots,p_d\}\subset \ZZ^k$ is a fundamental object since it
expresses the triangulations of conv$(A)$ via a polytope representation.
For any triangulation $T$ of conv$(A)$, define vector $\phi_T\in\RR^{d}$
with $i$-coordinate
\begin{equation}\label{phi_vec}
\phi_T(i)= \sum_{\sigma\in T \ |\ p_i\in\text{vtx}(\sigma)} \mbox{vol}(\sigma),
\end{equation}
summing over all simplices $\sigma$ of $T$ having $p_i$ as a vertex, where vtx$(\sigma)$ is the vertex set of simplex $\sigma$, and $i\in \{1,\dots,d\}$. Now the secondary polytope $\Sigma(\A)$, or just $\Sigma$, is defined as the convex hull of $\phi_T$ for all triangulations $T$.
A famous theorem of~\cite{GKZ}, which is also the central result
in~\cite{DeLRamSan}, states that
there is a bijection between the vertices of $\Sigma$ and 
the regular triangulations of $\conv(A)$.
This extends to a bijection between the face poset of $\Sigma$
and the poset of regular subdivisions of $\conv(A)$.
Moreover, $\Sigma$, although in ambient space
$\RR^{d}$, has actual dimension $\dim(\Sigma)=d-k-1$.

Let us now consider the Newton polytope of resultants, or
{\em resultant polytopes}, for which
optimization oracles provide today the only plausible approach
for their computation~\cite{EFKP12}.

Let us consider sets $A_0,\ldots,A_k \subset \ZZ^k$.
In the algebraic setting, these are the supports of $k+1$
polynomials in $k$ variables.
Let the Cayley set be defined by
$$
A:=\bigcup_{j=0}^{k} (A_{j} \times \{e_{j}\}) \subset \ZZ^{2k},
$$
where $e_0,\ldots,e_k$ form an affine basis of $\ZZ^k$.
Clearly, each point in $A$ corresponds to a unique point in some $A_i$.
The (regular) triangulations of $A$ are in bijective correspondence
with the (regular) fine mixed subdivisions of the Minkowski sum
$A_0+\cdots + A_k$~\cite{GKZ}.
Mixed subdivisions are those where all cells are Minkowski sums of 
convex hulls of subsets of the $A_i$.
A mixed subdivision is fine if, for every cell, the sum
of its summands' dimensions equals the dimension of the cell.

Let $d:=\sum_{j=0}^{k}|A_j|$, then given triangulation $T$ of $\conv(A)$, define vector $\rho_T\in\RR^{d}$ with $i$-coordinate 
\begin{equation}\label{rho_vec}
\rho_T(i) := \sum_{
  {i\text{-mixed }} {\sigma\in T}}
\mbox{vol}(\sigma), 
\end{equation}
where $i\in \{1,\dots,d\}$.
A simplex $\sigma$ is called {\em $i$-mixed} if it contains $p_i\in A_\ell$ for some $\ell\in\{1,\dots,k\}$ and exactly $2$ points from each $A_j$, where $j$ ranges over $\{0,1,\dots,k\}-\{ \ell\}$. 
The {\em resultant polytope} $R$ is defined as the convex hull of $\rho_T$ for all triangulations $T$. Similarly with the secondary polytope, it is in ambient space $\RR^d$ but has dimension $\dim(R)= d-2k-1$~\cite{GKZ}.
There is a surjection, i.e. many to one relation, from the regular triangulations of $\conv(A)$ to the vertices of $R$.

\begin{example}
Let $A_0=\{\{0\},\{2\}\},\ A_1=\{\{0\},\{1\},\{2\}\}$, then the Cayley set will be $A=\{\{0,0\},\{2,0\},\{0,1\},\{1,1\},\{2,1\}\}$. The 5 vertices of the secondary polytope $\Sigma(A)$ are computed using equation~(\ref{phi_vec}):
\begin{align*}
\phi(T_1)&=(2,4,2,0,4),\\
\phi(T_2)&=(4,2,4,0,2),\\
\phi(T_3)&=(4,2,3,2,1),\\
\phi(T_4)&=(3,3,1,4,1),\\
\phi(T_5)&=(2,4,1,2,3),
\end{align*}
and the 3 vertices of the resultant polytope $N(R)$ are computed using equation~(\ref{rho_vec}):
\begin{align*}
\rho(T_1)&=(0,2,0,0,2),\\
\rho(T_2)&=(2,0,2,0,0),\\
\rho(T_3)&=(2,0,2,0,0),\\
\rho(T_4)&=(1,1,0,2,0),\\
\rho(T_5)&=(0,2,0,0,2).
\end{align*}
Note that there are two pairs of triangulations that yield one resultant vertex each. Figure~\ref{sec_res_example} illustrates this example.
\end{example}

\begin{figure}
\includegraphics[width=\textwidth]{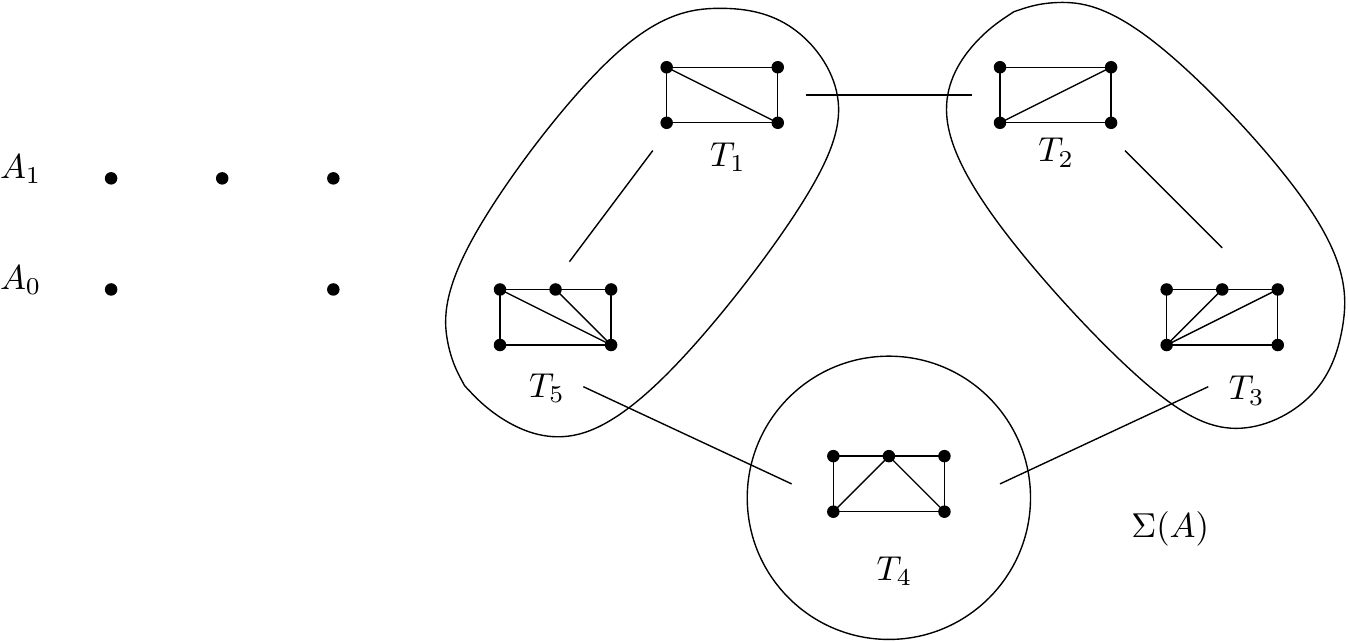}
\caption{Secondary and resultant polytopes.\label{sec_res_example}}
\end{figure}

We consider $k$ fixed because in practice it holds $k\ll d\ll n$,
where $n$ stands for the number of polytope vertices.
Note that $R$ is computed as a full-dimensional polytope in a space
of its intrinsic dimension~\cite{EFKP12} and this approach 
extends to $\Sigma$.

Computing the V-representation of $\Sigma$ and $R$ by the algorithm
in~\cite{EFKP12} is not 
total polynomial.
In fact, the complexity depends on the number of
polytope vertices and facets, but also on the number of
simplices in a triangulation of the polytope (see Proposition~\ref{BB}).
However, we show that Algorithm~\ref{Alg:edge_skeleton} 
computes $\Sigma$ and $R$ in oracle total polynomial-time.

\begin{lemma}\label{lem:R_encoding}
Both $\Sigma$ and $R$ are well-described polytopes. 
\end{lemma}

\begin{proof}
For the case of $\Sigma$, given $A\in\ZZ^{k}$, let $\enc{A}$ be its encoding length and $\alpha :=$ vol(conv($A$)).
It is clear that $\alpha=O(\enc{A}^k)$ and thus $\enc{\alpha}=O(k\enc{A})$. 
For each triangulation $T$ each coordinate of $\phi_T$ is upper bounded 
by $\alpha$, since the sum of the volumes of its adjacent simplices 
cannot exceed vol(conv($A$)). This bound is tight for the points 
$a\in A$ of a regular triangulation $T$ where the simplices containing
$a$ partition $\conv(A)$.
It follows that the encoding length of $\Sigma$ vertices 
is $\enc{\alpha}$ and thus $\enc{\Sigma}=4n^2\enc{\alpha}+d=O(dn^2\enc{A})$
by~Lemma~\ref{lem:encoding}.
Similarly, we bound the encoding length of $\rho_T$ which yields that $R$ is also a well-described polytope.
\end{proof}

In the sequel, we characterize the set of edge directions of $\Sigma$ and $R$. 
The edge directions of both $\Sigma,R$ can be computed by enumerating circuits 
of $A$. More specifically, circuit enumeration suffices to compute the 
{\em edge vectors}, i.e.\ both directions and lengths of the edges. 

We first give some fundamental definitions from combinatorial geometry. For a detailed 
presentation we recommend~\cite{DeLRamSan}. 
A {\it circuit} $C\subseteq A$ is a minimum affinely dependent subset of $A$. It holds that conv$(C)$ has exactly two triangulations $C_+,C_-$. 
The operation of switching from one triangulation to another is called \textit{flip}.
Triangulation $T$ of $A$, which equals $C_+$ when restricted on circuit $C$,
is {supported} on $C$ if, by flipping 
$C_+$ to $C_-$, we obtain another triangulation $T'$ of $A$.  
The dimension of a circuit is the dimension of its convex hull. 
If $A$ is in {\em generic position}, then
all circuits $C$ are full dimensional. 
Then all the edges of $\Sigma$ correspond to full dimensional circuits.
If $A$ is {\em not} in generic position, some edges may correspond to
lower-dimensional circuits.

In the case of $R$, where $A=\bigcup_{j=0}^k A_j$,
a circuit $C$ is called {\em cubical} if and only if $|C\cap A_j| \in \{0,2\}$,
$j= 0,\dots,k$. 
If $A$ is in {\em generic position}, all the edges of $R$ correspond to full dimensional cubical circuits~\cite{St94}.

\begin{lemma}\label{lem:res_edges}
Given $A\in\ZZ^{k}$ in generic position,
we compute the set of edge directions of $\Sigma$ in $O(d^{k+2})$.
Given $A\in\ZZ^{2k}$ in generic
position the set of edge directions of $R$ is computed in $O(d^{2k+2})$. 
In both cases, genericity of $A$ is checked within
the respective time complexity.
\end{lemma}

\begin{proof}
For $\Sigma$, we enumerate all $\binom{|A|}{k+2}$ circuits in
$O(d^{k+2})$, obtaining the set of all edge vectors.
Genericity of $A$ is established by checking whether all 
$\binom{|A|}{k}$ subsets, $k\in\{1,\dots,k+1\}$,
are independent. This is in $O(d^{k+1})$ for $k=O(1)$.

In the case of $R$, where $A=\bigcup_{j=0}^k A_j$,
a flip on $T$ is cubical if and only if it is
supported on a cubical circuit $C$.
In generic position, $|C |=2k+2$.
For those supporting cubical flips, $|C\cap A_j|=2$, $j= 0,\dots,k$.
Every edge $d_C$ of $R$ is supported on cubical flip $C$,
where $d_C(a)$ equals $\rho_{C_+}(a)-\rho_{C_-}(a)$,
if $a\in C$, and $0$ otherwise~\cite{St94}.
Given $A$, all such circuits are enumerated in
$\binom{|A|}{2k+2}=O(d^{2k+2})$; a better
bound is $O(t^{2k+2})$ if $t$ bounds $|A_j|,\ j=0,\dots,k$.
\end{proof}

\begin{lemma}{\rm\cite{EFKP12}} \label{EFKP12_prop}
For $k+1$ pointsets in $\ZZ^k$ of total cardinality $d$, optimization over
$R$ takes polynomial-time, when $k$ is fixed.
\end{lemma}

\begin{corollary}\label{cor:res_sec}
In total polynomial-time, we compute the edge-skeleton of $\Sigma\subset\RR^d$,
given $A\in\ZZ^{k}$ in generic position,
and the edge-skeleton of $R$, given $A\in\ZZ^{2k}$ in generic position. 
\end{corollary}

\begin{proof} 
Since by Lemma~\ref{lem:R_encoding} $\Sigma,R$ are well-bounded, 
optimization oracles are available by Lemma~\ref{EFKP12_prop} and 
the set of edge directions by Lemma~\ref{lem:res_edges}, 
the edge-skeletons of $\Sigma,R$ can be computed by 
Algorithm~\ref{Alg:edge_skeleton} in oracle total polynomial-time.
Moreover, since the optimization oracle is polynomial-time 
this yields a (proper) total polynomial-time algorithm for $\Sigma,R$. 
\end{proof}

Following Lemma~\ref{lem:res_edges}, for $\Sigma$, $R$
we also obtain their edge lengths.
This can lead to a more efficient edge-skeleton algorithm on the real RAM.

\begin{remark}
Our results readily extend to the Newton polytope of
discriminants, or discriminant polytopes. This follows from the fact that these polytopes can be written as signed Minkowski sums of secondary polytopes~\cite{GKZ}.
\end{remark}

\section{Concluding remarks}

We have presented the first total polynomial-time algorithm for computing the edge-skeleton
of a polytope, given an optimization oracle, and a set of directions 
that contains the polytope's edge directions.
Our algorithm yields the first (weakly) total polynomial-time algorithms 
for the edge-skeleton (and vertex enumeration) of
signed Minkowski sum, and resultant polytopes.

An open question is a {\em strongly} total polynomial-time algorithm 
for the edge-skeleton problem.
Another is to solve the edge-skeleton problem without edge directions; 
characterizations of edge directions for polytopes in H-representation 
are studied in~\cite{RothblumOnn_findedges}.
It is also interesting to investigate new classes of 
convex combinatorial optimization 
problems where our algorithm offers a polynomial-time algorithm.

\section{Acknowledgments}
The authors thank K.~Fukuda, C.~M\"uller, S.~Stich for discussions and
bibliographic suggestions.
Most of the work was done while Vissarion was at the University of Athens.
All authors are partially supported from project ``Computational
Geometric Learning'', which acknowledges the financial support of the Future and
Emerging Technologies programme within the 7th Framework Programme for
research of the European Commission, under FET-Open grant number: 255827.
I.Z.E.\ and V.F.\ are partially supported by the European Union (European Social Fund - ESF) and Greek national funds through the Operational Program ``Education and Lifelong Learning" of the National  Strategic Reference Framework (NSRF) - Research Funding Program: THALIS - UOA (MIS 375891).

\bibliographystyle{alpha}  
\bibliography{bibliography} 

\newcommand{\etalchar}[1]{$^{#1}$}
\begin{thebibliography}{DHO{\etalchar{+}}09}

\bibitem[ABD10]{Ardila10}
Federico Ardila, Carolina Benedetti, and Jeffrey Doker.
\newblock Matroid polytopes and their volumes.
\newblock {\em Discrete \& Computational Geometry}, 43(4):841--854, 2010.

\bibitem[ABS97]{AvisSeidel}
D.~Avis, D.~Bremner, and R.~Seidel.
\newblock How good are convex hull algorithms?
\newblock {\em Computational Geometry: Theory \& Appl.}, 7:265 -- 301, 1997.

\bibitem[AF92]{AvisF92}
D.~Avis and K.~Fukuda.
\newblock A pivoting algorithm for convex hulls and vertex enumeration of
  arrangements and polyhedra.
\newblock {\em Discrete {\&} Computational Geometry}, 8:295--313, 1992.

\bibitem[BDH09]{BoiDevHor09}
J.-D. Boissonnat, O.~Devillers, and S.~Hornus.
\newblock Incremental construction of the {Delaunay} triangulation and the
  {Delaunay} graph in medium dimension.
\newblock In {\em ACM Symp.\ on Comp.\ Geometry}, pages 208--216, 2009.

\bibitem[BFS90]{Billera1990155}
L.J. Billera, P.~Filliman, and B.~Sturmfels.
\newblock Constructions and complexity of secondary polytopes.
\newblock {\em Advances in Math.}, 83(2):155--179, 1990.

\bibitem[BL98]{Bussieck98}
M.R. Bussieck and M.E. L\"{u}bbecke.
\newblock The vertex set of a 0/1-polytope is strongly p-enumerable.
\newblock {\em Computational Geometry: Theory \& Appl.}, 11(2):103--109, 1998.

\bibitem[DHO{\etalchar{+}}09]{DeLoera20091569}
{J.A.} {De Loera}, R.~Hemmecke, S.~Onn, U.G. Rothblum, and R.~Weismantel.
\newblock Convex integer maximization via {Graver} bases.
\newblock {\em J. Pure \& Applied Algebra}, 213(8):1569--1577, 2009.

\bibitem[DRS10]{DeLRamSan}
{J.A.} {De Loera}, J.~Rambau, and F.~Santos.
\newblock {\em Triangulations: Structures for Algorithms and Applications},
  volume~25 of {\em Algorithms and Computation in Mathematics}.
\newblock Springer, 2010.

\bibitem[EFKP13]{EFKP12}
I.Z. Emiris, V.~Fisikopoulos, C.~Konaxis, and L.~Pe{\~n}aranda.
\newblock An oracle-based, output-sensitive algorithm for projections of
  resultant polytopes.
\newblock {\em Intern. J. Comp. Geom. Appl., Special Issue}, 23:397--423, 2013.

\bibitem[EPL82]{EPL82}
J.~Edmonds, W.R. Pulleyblank, and L.~Lov{\'a}sz.
\newblock Brick decompositions and the matching rank of graphs.
\newblock {\em Combinatorica}, 2(3):247--274, 1982.

\bibitem[Fuk04]{Fukuda04}
K.~Fukuda.
\newblock From the zonotope construction to the {Minkowski} addition of convex
  polytopes.
\newblock {\em J. Symbolic Computation}, 38, 2004.

\bibitem[FW05]{FukudaW05}
K.\ Fukuda and C.\ Weibel.
\newblock Computing all faces of the {M}inkowski sum of {V}-polytopes.
\newblock In {\em {C}anad.\ {C}onf.\ {C}omp.\ {G}eom.}, pages 253--256, 2005.

\bibitem[GH99]{Gritzmann99onthe}
P.~Gritzmann and A.~Hufnagel.
\newblock On the algorithmic complexity of {Minkowski's} reconstruction
  problem.
\newblock {\em J. London Math. Soc.}, 2:5--9, 1999.

\bibitem[GKZ94]{GKZ}
I.M. Gelfand, M.M. Kapranov, and A.V. Zelevinsky.
\newblock {\em Discriminants, Resultants and Multidimensional Determinants}.
\newblock Birkh\"{a}user, Boston, 1994.

\bibitem[GLS93]{GLS93}
M.~Gr{\"o}tschel, L.~Lov{\'a}sz, and A.~Schrijver.
\newblock {\em {Geometric Algorithms and Combinatorial Optimization}}, volume~2
  of {\em Algorithms and Combinatorics}.
\newblock Springer, 2nd corrected edition, 1993.

\bibitem[GS93]{Gritzmann93}
P.~Gritzmann and B.~Sturmfels.
\newblock Minkowski addition of polytopes: computational complexity and
  applications to {Gr{\"o}bner} bases.
\newblock {\em SIAM J. Discr. Math.}, 6(2):246--269, 1993.

\bibitem[Hug06]{Hug06}
P.~Huggins.
\newblock ib4e: A software framework for parametrizing specialized {LP}
  problems.
\newblock In A.~Iglesias and N.~Takayama, editors, {\em Mathematical Software -
  ICMS}, volume 4151 of {\em LNCS}, pages 245--247. Springer, 2006.

\bibitem[JKK02]{JoswigKaibelK02}
M.~Joswig, V.~Kaibel, and F.~K\"orner.
\newblock On the k-systems of a simple polytope.
\newblock {\em Israel J. Math.}, 129(1):109--117, 2002.

\bibitem[Kha79]{Kha1}
L.G.\ Khachiyan.
\newblock A polynomial algorithm in linear programming.
\newblock {\em Soviet Math. Doklady}, 20(1):191--194, 1979.

\bibitem[{Mal}14]{Malajovich14}
G.~{Malajovich}.
\newblock Computing mixed volume and all mixed cells in quermassintegral time.
\newblock {\em arXiv e-prints{ }}, 2014.

\bibitem[MC00]{MicCoo00}
T.\ Michiels and R.~Cools.
\newblock Decomposing the secondary cayley polytope.
\newblock {\em Discrete {\&} Computational Geometry}, 23:367--380, 2000.

\bibitem[McM71]{McMullen}
P.~McMullen.
\newblock The maximum numbers of faces of a convex polytope.
\newblock {\em Mathematika}, 17:179--184, 1971.

\bibitem[MII96]{Masada96}
T.~Masada, H.~Imai, and K.~Imai.
\newblock Enumeration of regular triangulations.
\newblock In {\em ACM Symp.\ on Comp.\ Geometry}, SoCG '96, pages 224--233,
  1996.

\bibitem[OR04]{OnnR04}
S.~Onn and U.G. Rothblum.
\newblock Convex combinatorial optimization.
\newblock {\em Discrete {\&} Computational Geometry}, 32(4):549--566, 2004.

\bibitem[OR07]{RothblumOnn}
S.~Onn and U.G. Rothblum.
\newblock The use of edge-directions and linear programming to enumerate
  vertices.
\newblock {\em J. Combin. Optim.}, 14:153--164, 2007.

\bibitem[Ore99]{discrim_vol}
S.Yu. Orevkov.
\newblock {The volume of the Newton polytope of a discriminant.}
\newblock {\em Russ. Math. Surv.}, 54(5):1033--1034, 1999.

\bibitem[ORT05]{RothblumOnn_findedges}
S.~Onn, U.G. Rothblum, and Y.~Tangir.
\newblock Edge-directions of standard polyhedra with applications to network
  flows.
\newblock {\em J. Global Optimization}, 33(1):109--122, September 2005.

\bibitem[PR03]{TOPCOM2}
J.~Pfeifle and J.~Rambau.
\newblock Computing triangulations using oriented matroids.
\newblock In Michael Joswig and Nobuki Takayama, editors, {\em Algebra,
  Geometry and Software Systems}, pages 49--75. Springer Berlin Heidelberg,
  2003.

\bibitem[{Sch}93]{Schneider93}
Rolf {Schneider}.
\newblock {\em Convex bodies: the {Brunn-Minkowski} theory}.
\newblock Cambridge University Press, 1993.

\bibitem[Sta04]{stanley2004combinatorics}
R.P. Stanley.
\newblock {\em Combinatorics and Commutative Algebra}.
\newblock Combinatorics and Commutative Algebra. Birkh{\"a}user Boston, 2004.

\bibitem[Stu94]{St94}
B.\ Sturmfels.
\newblock On the {Newton} polytope of the resultant.
\newblock {\em J.\ Algebraic Combin.}, 3:207--236, 1994.

\bibitem[Zie95]{Ziegler}
G.M. Ziegler.
\newblock {\em Lectures on Polytopes}.
\newblock Springer, 1995.

\end{thebibliography}

\end{document}